\documentclass[11pt,a4paper]{article}
\usepackage{amsmath,amssymb,amsthm,graphicx,braket,bm,multirow,authblk,amsthm,dcolumn,latexsym,subcaption,mathtools,cite}
\usepackage[normalem]{ulem}
\usepackage[a4paper]{geometry}
\DeclareSymbolFontAlphabet{\amsmathbb}{AMSb}%

\usepackage[unicode=true,
bookmarks=true,bookmarksopen=false,
breaklinks=false,pdfborder={0 0 0},colorlinks=true]
{hyperref}

\usepackage{xcolor}
\definecolor{cblue}{rgb}{0.16, 0.32, 0.75}
\definecolor{cred}{rgb}{0.7, 0.11, 0.11}
\hypersetup{%
	,linkcolor=cred
	,citecolor=cblue
	,urlcolor=black
}

\def\oper{{\mathchoice{\rm 1\mskip-4mu l}{\rm 1\mskip-4mu l}
		{\rm 1\mskip-4.5mu l}{\rm 1\mskip-5mu l}}}

\renewcommand{\i}{\mathrm{i}}
\newcommand{\e}{\mathrm{e}}
\newcommand{\g}{\mathrm{g}}
\newcommand{\ee}{\mathrm{ee}}
\renewcommand{\gg}{\mathrm{gg}}
\newcommand{\eg}{\mathrm{eg}}
\renewcommand{\ge}{\mathrm{ge}}
\newcommand{\hilb}{\mathcal{H}}
\newcommand{\hilbe}{\mathcal{H}_\e}
\newcommand{\hilbg}{\mathcal{H}_\g}
\newcommand{\tr}{\operatorname{Tr}}
\newcommand{\ketbra}[2]{| #1 \rangle\!\langle #2 | }
\newcommand{\bh}{\mathcal{B}(\mathcal{H})}
\newcommand{\bhe}{\mathcal{B}(\mathcal{H}_\e)}
\newcommand{\bhg}{\mathcal{B}(\mathcal{H}_\g)}

\renewcommand{\Re}{\mathop{\mathrm{Re}}}

\newcommand{\X}{\mathsf{X}}

\newtheorem{theorem}{Theorem}[section]
\newtheorem{proposition}{Proposition}[section]
\newtheorem{corollary}{Corollary}[section]

\newtheorem{lemma}{Lemma}[section]
\theoremstyle{remark}\newtheorem{remark}{Remark}[section]
\theoremstyle{remark}\newtheorem{example}{Example}[section]

\newcommand{\blockmatrix}[4]{\left[
	\begin{array}{c|c}
		#1 & #2 \\\hline #3 & #4
	\end{array}\right]
}

\begin{document}	
	\title{
		\textbf{Excitation-damping quantum channels}
	}
	
	\author[$\hspace{0cm}$]{Davide Lonigro$^{1,2,}$\footnote{davide.lonigro@ba.infn.it}}
	\affil[$1$]{\small Dipartimento di Fisica and MECENAS, Universit\`{a} di Bari, I-70126 Bari, Italy}
	\affil[$2$]{\small INFN, Sezione di Bari, I-70126 Bari, Italy}
	
	\author[$\hspace{0cm}$]{Dariusz Chru\'sci\'nski$^{3,}$\footnote{darch@fizyka.umk.pl}}
	\affil[$3$]{\small Institute of Physics, Faculty of Physics, Astronomy and Informatics, Nicolaus Copernicus University, Grudziadzka 5/7, 87-100 Toru\'n, Poland}
	
	\maketitle
	\vspace{-0.5cm}	
	
	\begin{abstract}		
		We study a class of quantum channels describing a quantum system, split into the direct sum of an excited and a ground sector, undergoing a one-way transfer of population from the former to the latter; this construction, which provides a generalization of the amplitude-damping qubit channel, can be regarded as a way to upgrade a trace non-increasing quantum operation, defined on the excited sector, to a possibly trace preserving operation on a larger Hilbert space. We provide necessary and sufficient conditions for the complete positivity of such channels, and we also show that complete positivity is equivalent to simple positivity whenever the ground sector is one-dimensional. Finally, we examine the time-dependent scenario and characterize all CP-divisible channels and Markovian semigroups belonging to this class.
	\end{abstract}
	
	\maketitle
	
	\section{Introduction}
	
	Quantum channels, represented by completely positive and trace preserving (CPTP) linear maps, are nowadays key objects in modern quantum information theory~\cite{QIT}. They represent physically legitimate state transformations which are consistent with the structure of composite quantum systems, that is, the tensor product of two channels $\Phi_A$ and $\Phi_B$  operating separately on the systems ``$A$'' and ``$B$'' defines a quantum channel $\Phi_A \otimes \Phi_B$ operating on the composite ``$AB$'' system. This property is in general violated by maps which are only positive but not completely positive: in such a case, even if $\Phi_A$ and $\Phi_B$ safely transform states of systems ``$A$'' and ``$B$'', the tensor product needs not be positive and hence, in general, fails to properly transform entangled states of the ``$AB$'' system.
	
	A quantum channel provides a powerful generalization of a unitary map $\rho \to \mathcal{U}(\rho) = U\rho U^\dagger$, and hence it allows one to properly represent a quantum evolution beyond the Schr\"odinger unitary scenario. Any physically legitimate evolution of an open quantum system~\cite{Open1,Open2} may be represented by a dynamical map, i.e.~a family of quantum channels $\{\Lambda_t\}_{t \geq 0}$; in particular, a Markovian semigroup is represented by $\Lambda_t = \e^{t \mathcal{L}}$, where $\mathcal{L}$ stands for the Gorini--Kossakowski--Lindblad--Sudarshan (GKLS) generator~\cite{GKS,L} (cf.~\cite{ALICKI} for a detailed exposition and~\cite{40-GKLS} for historical remarks). Recently, the evolution of open systems beyond the Markovian semigroup scenario has been attracting considerable attention both from a theoretical and experimental point of view (cf.~\cite{NM1,NM2,NM3} and~\cite{Piilo-I,Piilo-II}); in particular, various inequivalent concepts and measures of non-Markovianity have been proposed. The whole hierarchy of different approaches was recently analyzed in great detail in~\cite{NM4} (see also the recent tutorial~\cite{Kavan}).
	
	In this paper we discuss a class of quantum channels which provides a generalization of the well-known amplitude-damping qubit channel as well as its multilevel version studied in~\cite{Davide-1}. We will denote such channels as {\em excitation-damping quantum channels}. The essential ingredient of the construction is a splitting of the Hilbert space $\hilb$ of the system as the direct sum of an ``excited'' sector $\hilbe$ and a ``ground'' sector $\hilbg$, i.e.~$\hilb = \hilbe \oplus \hilbg$. Excitation-damping channels will be defined in such a way to involve a one-way population transfer from the excited sector to the ground one, thus justifying our notation, as well as a modulation of the coherence between the two sectors.
	
	Importantly, as we will show, such a construction may also be considered as a way to upgrade a trace non-increasing quantum operation $\phi:\bhe\rightarrow\bhe$ to a possibly trace preserving operation $\mathsf{\Phi}$, that is, a true quantum channel, on a larger Hilbert space $\hilb$ obtained from $\hilbe$ by adding additional degrees of freedom---a ground sector. This procedure is well-known when one deals with decaying unstable systems and enlarges the original Hilbert space by states which represent the decay products of unstable states~\cite{Caban,Beatrix,Benatti,Kordian,Beatrix2}. The construction we propose provides a substantial generalization of this approach. 
	
	The paper is organized as follows:
	\begin{itemize}
		\item In Section~\ref{I} we study the mathematical properties of excitation-damping maps and, in particular,  completely positive and trace preserving excitation-damping maps (excitation-damping quantum channels)~\cite{Paulsen,Stormer}. After studying the invertibility of such maps (Proposition~\ref{prop:inv}), we find necessary and sufficient conditions for their complete positivity (Theorem~\ref{thm:cp}); besides, in the case of a one-dimensional ground state, we also find necessary and sufficient conditions for positivity (Proposition~\ref{prop:p}), which turns out to be equivalent to complete positivity.
		\item In Section~\ref{II} we employ the results in the previous section to analyze dynamical maps corresponding to time-dependent excitation-damping quantum channels; in this framework, we characterize all possible dynamical semigroups (Theorem~\ref{thm:semigroups}) as well as invertible CP-divisible maps (Theorem~\ref{thm:cpdiv}) in this class, fully characterizing the corresponding generators.
	\end{itemize}
	Final considerations are outlined in Section~\ref{III}.
	
	\section{Properties of excitation-damping maps}\label{I}
	
	\subsection{Generalities}	
	Let $\hilbe$, $\hilbg$ two Hilbert spaces with dimensions $d_\e,d_\g<\infty$, and $\hilb=\hilbe\oplus\hilbg$. The most general element $\mathsf{X}\in\bh$ can be thus partitioned as
	\begin{equation}\label{eq:part}
		\mathsf{X}=\blockmatrix{X_\ee}{X_\eg}{X_\ge}{X_\gg},\qquad X_{ss'}\in\mathcal{B}(\hilb_{s'},\hilb_s),\;\;s,s'\in\{\e,\g\}.
	\end{equation}
	Given two maps $\phi:\bhe\rightarrow\bhe$, $\omega:\bhe\rightarrow\bhg$, an operator $B\in\bhe$, and $\gamma\geq0$, we consider the map $\mathsf{\Phi}:\bh\rightarrow\bh$ acting on all $\X\in\bh$, partitioned as in Eq.~\eqref{eq:part}, via
	\begin{equation}\label{eq:phi0}
		\mathsf{\Phi}(\X)=\blockmatrix{\phi(X_\ee)}{BX_\eg}{X_\ge B^\dag}{\gamma X_\gg+\omega(X_\ee)}.
	\end{equation}
	We will refer to such maps as \textit{excitation-damping maps}. As pointed out in the introduction, they involve a one-way exchange of population from the excited sector to the ground one. Clearly, such maps are trace preserving if and only if
	\begin{equation}\label{eq:tp0}
		\gamma=1\qquad\text{and}\quad\tr\phi(X_\ee)+\tr\omega(X_\ee)=\tr X_\ee\;\text{for all}\;X_\ee\in\bhe.
	\end{equation}	
	While we are largely interested in the case $\gamma=1$, for the purposes of next section it will be convenient to let $\gamma$ be arbitrary.	
	\begin{example}\label{ex:tp}
		We shall often consider a subclass of manifestly trace preserving excitation-damping maps. Take $\gamma=1$ and, given a state $\Omega\in\bhg$ with $\tr\Omega=1$, let $\omega:\bhe\rightarrow\bhg$ defined as such:\begin{equation}
			X_\ee\in\bhe\mapsto	\omega(X_\ee)=\tr\left[X_\ee-\phi(X_\ee)\right]\Omega,
		\end{equation}
		so that
		\begin{equation}\label{eq:phi2}
			\mathsf{\Phi}(\X)=\blockmatrix{\phi(X_\ee)}{BX_\eg}{X_\ge B^\dag}{X_\gg+\tr\left[X_\ee-\phi(X_\ee)\right]\Omega}.
		\end{equation}
		Clearly, this map is trace preserving independently of the particular choice of $\phi$ and $B$; such a construction is therefore a simple way to create a trace preserving map on $\bh$ starting from a generally trace non-preserving map on $\bhe$.
		
		This family of excitation-damping maps includes some examples commonly found in the literature. For example, in the qubit case ($d_\e=d_\g=1$), it includes the following maps:
		\begin{equation}\label{eq:qubit}
			\mathsf{\Phi}(\X)=\blockmatrix{|a|^2x_\ee}{bx_\eg}{b^*x_\ge}{x_\gg+\left(1-|a|^2\right)x_\ee},
		\end{equation}
		with $a,b\in\mathbb{C}$; for $b=a$ this channel corresponds to an amplitude-damping channel, while for $|a|=1$ it reduces to a phase-damping channel. More generally, if $d_\g=1$ and $\phi(X_\ee)=AX_\ee A^\dag$, with $A\in\bhe$, we get
		\begin{equation}
			\mathsf{\Phi}(\X)=\blockmatrix{AX_\ee A^\dag}{BX_\eg}{X_\ge B^\dag}{x_\gg+\tr[(\oper_\e-A^\dag A)X_\ee]},
		\end{equation}
		which, for $B=A$, reduces to the multilevel generalization of the amplitude-damping channel studied in~\cite{Davide-1}. All such models emerge naturally when taking into account the reduced dynamics induced by atom-field interactions.
	\end{example}
	
	Interestingly, excitation-damping maps are invertible under minimal assumptions:
	\begin{proposition}\label{prop:inv}
		The map $\mathsf{\Phi}:\bh\rightarrow\bh$ as in Eq.~\eqref{eq:phi0} is invertible if and only if $\gamma\neq0$ and both $\phi,B$ are invertible. In such a case, for all $\X\in\bh$
		\begin{equation}\label{eq:phinv}
			\mathsf{\Phi}^{-1}(\X)=\blockmatrix{\phi^{-1}(X_\ee)}{B^{-1}X_\eg}{X_\ge B^{\dag-1}}{\gamma^{-1}\left(X_\gg-\omega\phi^{-1}(X_\ee)\right)}.
		\end{equation}
	\end{proposition}
	\begin{proof}
		If $\gamma\neq0$ and $\phi^{-1},B^{-1}$ exist, then an immediate computation shows that the map in Eq.~\eqref{eq:phinv} is the inverse of $\mathsf{\Phi}$. If any of said conditions fail, one can immediately construct counterexamples; for instance, if $\phi$ is not invertible, i.e.~there exist two distinct $X_\ee,X_\ee'\in\bhe$ such that $\phi(X_\ee)=\phi(X_\ee')$, then
		\begin{equation}
			\mathsf{\Phi}\blockmatrix{X_\ee}{0}{0}{0}=\blockmatrix{\phi(X_\ee)}{0}{0}{0}=\blockmatrix{\phi(X_\ee')}{0}{0}{0}=\mathsf{\Phi}\blockmatrix{X_\ee'}{0}{0}{0},
		\end{equation}
		and analogously in the other cases.
	\end{proof}
	
	\subsection{Complete positivity, general case}
	A natural question is whether the excitation-damping construction, under suitable assumptions, renders $\mathsf{\Phi}$ completely positive (and, in particular, CPTP) when $\phi$ is completely positive. As we will see, a complete characterization of all choices of $\phi$ and $B$ yielding a completely positive map can be indeed reached.
	
	\begin{theorem}\label{thm:cp}
		Let $\phi:\bhe\rightarrow\bhe$, $\omega:\bhe\rightarrow\bhg$, $B\in\bhe$ and $\gamma\geq0$; let $\mathsf{\Phi}$ as in Eq.~\eqref{eq:phi0}. The following statements are equivalent:
		\begin{itemize}
			\item[(i)] $\mathsf{\Phi}$ is completely positive;
			\item[(ii)] the map $\omega$ is completely positive, and one of the following conditions hold:
			\begin{itemize}
				\item[$\bullet$]$\gamma=0$ and $B=0$;
				\item[$\bullet$] $\gamma>0$, and the map $\phi-\frac{1}{\gamma}B(\cdot)B^\dag$ is completely positive;
			\end{itemize}
			\item[(iii)] the map $\omega$ is completely positive, and there exist $\{A_\mu\}_{\mu=1,\dots,r}\subset\bhe$ and $\{\beta_\mu\}_{\mu=1,\dots,r}\subset\mathbb{C}$, with $r\leq d_\e^2$, such that
			\begin{equation}
				\phi=\sum_{\mu=1}^rA_\mu(\cdot)A_\mu^\dag,\qquad B=\sum_{\mu=1}^r\beta_\mu A_\mu\;\;\text{with}\;\;\sum_{\mu=1}^r|\beta_\mu|^2\leq\gamma.
			\end{equation}
		\end{itemize}
		Finally, the map is trace preserving if and only if Eq.~\eqref{eq:tp0} holds.
	\end{theorem}
	Theorem~\ref{thm:cp} clarifies precisely the conditions under which $\phi$, $\omega$ and $B$ must be chosen so that $\mathsf{\Phi}$ is a CPTP map. Necessarily, $\phi$ and $\omega$ themselves must be completely positive; furthermore, the operator $B$ governing the off-diagonal part of $\X$ must be ``sufficiently small" with respect to $\phi$, in the sense that $\phi-\gamma^{-1}B(\cdot)B^\dag$ is still completely positive; as it turns out, this happens if and only if $B$ belongs to the set
	\begin{equation}\label{eq:bphi}
		\mathcal{B}_{\phi,\gamma}=\left\{B=\sum_{\mu=1}^r\beta_\mu A_\mu\in\bhe,\;\sum_{\mu=1}^r|\beta_\mu|^2\leq\gamma\right\},
	\end{equation}
	that is, the (complex) ball of radius $\sqrt{\gamma}$ spanned by the Kraus operators associated with $\phi$. In this sense, for a fixed $\phi$ with Kraus rank $r$, Theorem~\ref{thm:cp} guarantees the existence of a $r$-dimensional ``ball of completely positive maps", whose radius coincides with the square root of $\gamma$. Notice that, since the Kraus operators associated with $\phi$ are unique up to a unitary transformation, said ball is indeed independent of the particular choice of Kraus operators.
	
	\begin{remark}
		Recall that a positive map $\phi:\bhe\rightarrow\bhe$ is said to be \textit{trace non-increasing} if, for all $X_\ee\succeq0$, the inequality $\tr\phi(X_\ee)\leq\tr X_\ee$ holds. It is easy to see that such a property holds if and only if, given any orthonormal basis $\{\ket{\e_j}\}_{j=1,\dots,d_\e}\subset\hilbe$, the matrix with $(j,\ell)$th element
		\begin{equation}
			\delta_{j\ell}-\tr\phi\left(\ketbra{\e_j}{\e_\ell}\right),\qquad j,\ell=1,\dots,d_\e
		\end{equation}
		is positive semidefinite. Furthermore, $\phi$ is completely positive and trace non-increasing if and only if it admits a Kraus representation
		\begin{equation}\label{eq:krausni}
			\phi=\sum_{\mu=1}^r A_\mu(\cdot)A_\mu^\dag\quad\text{with}\quad\sum_{\mu=1}^r A_\mu^\dag A_\mu\preceq\oper_\e.
		\end{equation}
		Clearly, Eq.~\eqref{eq:tp0} implies that, in order an excitation-damping map $\mathsf{\Phi}$ to be trace preserving, $\phi$ must necessarily be trace non-increasing. Consequently, as anticipated, constructing an excitation-damping channel $\mathsf{\Phi}$ may be regarded as a way to ``promote'' the trace non-increasing map $\phi$ to a legitimate quantum channel.
	\end{remark}
	
	We remark that such a simple characterization of all choices of $\phi$ and $B$ rendering the corresponding excitation-damping map $\mathsf{\Phi}$ completely positive does not have, in general, a counterpart for positive maps; an important exception, which will be examined later on, is the case $d_\g=1$ (cf.~Subsection~\ref{subsec:dg1}).
	
	The remainder of this section will be devoted to proving Theorem~\ref{thm:cp}, and to discuss particular cases.
	\begin{lemma}\label{lemma:ball}
		Let $\phi:\bhe\rightarrow\bhe$ and $B\in\bhe$. The following properties are equivalent:
		\begin{itemize}
			\item [(i)] the map $\phi-B(\cdot)B^\dag$ is completely positive;
			\item [(ii)] there exist $\{A_\mu\}_{\mu=1,\dots,r}\subset\bhe$ and $\{\beta_\mu\}_{\mu=1,\dots,r}\subset\mathbb{C}$, with $r\leq d_\e^2$, such that
			\begin{equation}\label{eq:kraus}
				\phi=\sum_{\mu=1}^rA_\mu(\cdot)A_\mu^\dag
			\end{equation}
			and
			\begin{equation}\label{eq:ti}
				B=\sum_{\mu=1}^r\beta_\mu A_\mu,\quad\text{with}\quad\sum_{\mu=1}^r|\beta_\mu|^2\leq1.
			\end{equation}
		\end{itemize}
	\end{lemma}
	\begin{proof}
		(i)$\implies$(ii) Let $\phi-B(\cdot)B^\dag$ be completely positive: since $B(\cdot)B^\dag$ is completely positive by construction, $\phi$ is completely positive as well. Therefore, there exist two families $\{A_\mu\}_{\mu=1,\dots,r}$ and $\{A'_\nu\}_{\mu=1,\dots,s}$, with $r,s\leq d_\e^2$, such that
		\begin{eqnarray}
			\phi&=&\sum_{\mu=1}^r A_\mu(\cdot)A_\mu^\dag;\\
			\phi-B(\cdot)B^\dag&=&\sum_{\nu=1}^{s} A'_\nu(\cdot)A_\nu^{'\dag};
		\end{eqnarray}
		which implies
		\begin{equation}
			\sum_{\mu=1}^r A_\mu(\cdot)A_\mu^\dag=\sum_{\nu=1}^{s} A'_\nu(\cdot)A_\nu^{'\dag}+B(\cdot)B^\dag.
		\end{equation}
		Therefore, $\{A_\mu\}_{\mu=1,\dots,r}$ and $\{B\}\cup\{A'_\nu\}_{\nu=1,\dots,r'}$ are two families of Kraus operators representing the same map $\phi$. Consequently, by appending zeros to the smaller set of operators, there exists a unitary matrix $\left(u_{\mu\nu}\right)_{\mu,\nu=1,\dots,m}$, with $m=\max\{r,s+1\}$, which transforms one set into the other. In particular, in all cases
		\begin{equation}
			B=\sum_{\mu=1}^r u_{1\mu}A_\mu,
		\end{equation}
		since either $r=m$ or the remaining terms in the sum above are zero. Since the rows of an unitary matrix have unit norm, in the first case $\sum_\mu|u_{1\mu}|^2=1$, while in the second case $\sum_\mu|u_{1\mu}|^2\leq1$.
		
		(ii)$\implies$(i) Let $\phi$ and $B$ as in Eqs.~\eqref{eq:kraus}--\eqref{eq:ti}. Then
		\begin{equation}
			\phi-B(\cdot)B^\dag=\sum_{\mu,\nu=1}^r\left[\delta_{\mu\nu}-\beta_\mu\beta_\nu^*\right]A_\mu(\cdot)A_\nu^\dag,
		\end{equation}
		implying that $\phi-B(\cdot)B$ is completely positive if and only if the matrix $\left(\delta_{\mu\nu}-\beta_\mu\beta_\nu\right)$ is positive semidefinite, which happens if and only if $\sum_\mu|\beta_\mu|^2\leq1$.	
	\end{proof}
	
	\begin{proof}[Proof of Theorem~\ref{thm:cp}]
		The case $\gamma=0$ is obvious, so take $\gamma>0$. Let $\{\ket{\e_j}\}_{j=1,\dots,d_\e}\subset\hilbe$, $\{\ket{\g_a}\}_{a=1,\dots,d_\g}\subset\hilbg$ two orthonormal bases, and define the (unnormalized) maximally entangled vector
		\begin{equation}
			\ket{\Psi}=\sum_{j=1}^{d_\e}\ket{\e_j,\e_j}+\sum_{a=1}^{d_\g}\ket{\g_a,\g_a}\in\hilb\otimes\hilb,
		\end{equation}
		where we use the shorthand $\ket{u,v}\equiv\ket{u}\otimes\ket{v}$. Choi's theorem on completely positive maps~\cite{Choi-75} ensures that $\mathsf{\Phi}$ is completely positive if and only if the operator
		\begin{equation}\label{eq:choi}
			\mathsf{C}_{\mathsf{\Phi}}=\left(\mathsf{\Phi}\otimes\mathsf{id}\right)(\ketbra{\Psi}{\Psi})\in\bh\otimes\bh\simeq\mathcal{B}(\hilb\otimes\hilb)
		\end{equation}
		is positive semidefinite; similarly, defining
		\begin{equation}
			\ket{\Psi_\e}=\sum_{j=1}^{d_\e}\ket{\e_j,\e_j}\in\hilb_\e\otimes\hilb_\e,\qquad\ket{\Psi_\g}=\sum_{a=1}^{d_\g}\ket{\g_a,\g_a}\in\hilb_\g\otimes\hilb_\g,
		\end{equation}
		then $\phi$ and $\omega$ are completely positive if and only if the operators
		\begin{eqnarray}
			C_{\phi}&=&\left(\phi\otimes\mathrm{id}\right)\left(\ketbra{\Psi_\e}{\Psi_\e}\right)\in\bhe\otimes\bhe\simeq\mathcal{B}(\hilbe\otimes\hilbe);\\
			C_{\omega}&=&\left(\omega\otimes\mathrm{id}\right)\left(\ketbra{\Psi_\e}{\Psi_\e}\right)\in\bhg\otimes\bhe\simeq\mathcal{B}(\hilbg\otimes\hilbe),
		\end{eqnarray}
		are, respectively, positive semidefinite. Now, by Eq.~\eqref{eq:choi} and the definition of $\mathsf{\Phi}$, a simple calculation yields
		\begin{equation}
			\mathsf{C}_{\mathsf{\Phi}}=\left[
			\begin{array}{c|c|c|c}
				C_\phi&\cdot&\cdot&(B\otimes\oper_\e)\ketbra{\Psi_\e}{\Psi_\g}\\\hline
				\cdot&\cdot&\cdot&\cdot\\\hline
				\cdot&\cdot&C_\omega&\cdot\\\hline
				\ketbra{\Psi_\g}{\Psi_\e}(B\otimes\oper_\e)^\dag&\cdot&\cdot&\gamma\ketbra{\Psi_\g}{\Psi_\g}
			\end{array}
			\right],
		\end{equation}
		where the partition
		\begin{eqnarray}
			\mathcal{B}(\hilb\otimes\hilb)\simeq\bigoplus_{r,r',s,s'=\e,\g}\mathcal{B}(\hilb_s\otimes\hilb_{s'},\hilb_r\otimes\hilb_{r'})
		\end{eqnarray}
		has been employed.
		
		The positive semidefiniteness of $\mathsf{C}_{\sf \Phi}$ can be characterized by means of the Schur complement~\cite{Bhatia,Matrices-2}. For $\gamma>0$ the operator $\gamma\ketbra{\Psi_\g}{\Psi_\g}$ admits a positive semidefinite generalized inverse $(1/\gamma d_\g^2)\ketbra{\Psi_\g}{\Psi_\g}$; a direct application of the Schur complement then shows that $\mathsf{C}_{\mathsf{\Phi}}$ is positive semidefinite if and only if
		\begin{equation}
			\left[
			\begin{array}{c|c|c}
				C_\phi-\frac{1}{\gamma}\left(B\otimes\oper_\e\right)\ketbra{\Psi_\e}{\Psi_\e}\left(B\otimes\oper_\e\right)^\dag&\cdot&\cdot\\\hline
				\cdot&\cdot&\cdot\\\hline
				\cdot&\cdot&C_\omega
			\end{array}
			\right]\succeq0,
		\end{equation}
		which clearly happens if and only if
		\begin{equation}\label{eq:chois}
			C_\phi-\frac{1}{\gamma}\left(B\otimes\oper_\e\right)\ketbra{\Psi_\e}{\Psi_\e}\left(B\otimes\oper_\e\right)^\dag\succeq0\quad\text{and}\quad C_\omega\succeq0.
		\end{equation}
		The latter condition is equivalent to the complete positivity of $\omega$, while the former condition is equivalent to the complete positivity of $\phi-\frac{1}{\gamma}B(\cdot)B^\dag$, since the first term in Eq.~\eqref{eq:chois} is indeed the Choi state of said map.
	\end{proof}
	
	\begin{remark}
		We can also find an explicit Kraus representation for $\mathsf{\Phi}$ (which, incidentally, provides an alternative proof of the implication (iii)$\implies$(i) of Theorem~\ref{thm:cp}). Let $\phi:\bhe\rightarrow\bhe$ and $\omega:\bhe\rightarrow\bhg$ completely positive; then they admit Kraus representations
		\begin{eqnarray}
			\phi&=&\sum_{\mu=1}^rA_\mu(\cdot)A_\mu^\dag,\qquad\{A_\mu\}_{\mu=1,\dots,r}\subset\bhe;\\
			\omega&=&\sum_{\nu=1}^{s}Q_\nu(\cdot)Q_\nu^\dag,\qquad\{Q_\nu\}_{\nu=1,\dots,s}\subset\mathcal{B}\left(\hilbe,\hilbg\right).
		\end{eqnarray}
		Take $B=\sum_{\mu=1}^r\beta_\mu A_\mu$ for some $\{\beta_\mu\}_{\mu=1,\dots,r}$ with $\sum_{\mu=1}^r|\beta_\mu|^2\leq\gamma$. We define a family of operators $\{\mathsf{A}_\mu\}_{\mu=0,1,\dots,r+s}$ by
		\begin{eqnarray}
			\mathsf{A}_0&=&\blockmatrix{0}{0}{0}{(\gamma-\sum_\mu|\beta_\mu|^2)^{1/2}\,\oper_\g};\\
			\mathsf{A}_\mu&=&\blockmatrix{A_\mu}{0}{0}{\beta_\mu^*\oper_\g},\;\mu=1,\dots,r;\\
			\mathsf{A}_{r+\nu}&=&\blockmatrix{0}{0}{Q_\nu}{0},\;\nu=1,\dots,s;
		\end{eqnarray}
		then an immediate computation shows that
		\begin{equation}
			\mathsf{\Phi}=\sum_{\mu=0}^{r+s}\mathsf{A}_\mu(\cdot)\mathsf{A}_\mu^\dag.
		\end{equation}
	\end{remark}
	
	We conclude this subsection by stating Theorem~\ref{thm:cp} for the subclass of manifestly trace preserving maps given in Example~\ref{ex:tp}, cf.~Eq.~\eqref{eq:phi2}.
	\begin{corollary}\label{coroll:cp}
		Let $\phi:\bhe\rightarrow\bhe$ and $B\in\bhe$; let $\mathsf{\Phi}$ as in Example~\ref{ex:tp}. The following statements are equivalent:
		\begin{itemize}
			\item[(i)] $\mathsf{\Phi}$ is completely positive and trace preserving;
			\item[(ii)] the map $\phi-B(\cdot)B^\dag$ is completely positive, and $\phi$ is trace non-increasing;
			\item[(iii)] there exist $\{A_\mu\}_{\mu=1,\dots,r}\subset\bhe$ and $\{\beta_\mu\}_{\mu=1,\dots,r}\subset\mathbb{C}$, with $r\leq d_\e^2$, such that
			\begin{equation}
				\phi=\sum_{\mu=1}^rA_\mu(\cdot)A_\mu^\dag,\qquad\sum_{\mu=1}^rA_\mu^\dag A_\mu\preceq\oper_\e
			\end{equation}
			and
			\begin{equation}\label{eq:b}
				B=\sum_{\mu=1}^r\beta_\mu A_\mu,\quad\text{with}\quad\sum_{\mu=1}^r|\beta_\mu|^2\leq1.
			\end{equation}
		\end{itemize}
	\end{corollary}
	
	\begin{proof}
		The map $\mathsf{\Phi}$ is trace preserving by construction. Besides, it readily follows from Theorem~\ref{thm:cp} with
		\begin{equation}
			\omega(X_\ee)=\tr\left[X_\ee-\phi(X_\ee)\right]\Omega
		\end{equation}
		that the latter map is completely positive if and only if the form
		\begin{equation}
			X_\ee\in\bhe\mapsto \tr\left[X_\ee-\phi(X_\ee)\right]\in\mathbb{C}
		\end{equation}
		is positive; but this condition is clearly equivalent to $\phi$ being a trace non-increasing map.
	\end{proof}
	
	\subsection{Case \texorpdfstring{$d_\g=1$}{dg=1}}\label{subsec:dg1}	
	We have provided a complete characterization of all completely positive excitation-damping maps for arbitrary values of the dimensions $d_\e$, $d_\g$ of both sectors of the Hilbert space $\hilb$. As previously remarked, an analogous characterization of positive maps in the general case is generally much more difficult.
	
	We will now particularize our discussion to the ``minimal" case $d_\g=1$, that is, the case of a one-dimensional ground sector to $\hilbe$. Choosing any $\ket{\g}\in\hilbg$ with unit norm, the most general $\X\in\bh$ can be thus written as
	\begin{equation}
		\X=\blockmatrix{X_\ee}{\ketbra{\xi_\e}{\g}}{\ketbra{\g}{\xi_\e}}{x_\gg\ketbra{\g}{\g}}\simeq\blockmatrix{X_\ee}{\ket{\xi_\e}}{\bra{\xi_\e}}{x_\gg}
	\end{equation}
	for some $\ket{\xi_\e}\in\hilbe$ and $x_\gg\in\mathbb{C}$, where we have applied the obvious isomorphism $\hilbg\simeq\mathbb{C}$ which will be let understood hereafter. With this representation, given $\phi:\bhe\rightarrow\bhe$ and $B\in\bhe$, the map $\mathsf{\Phi}$ in Eq.~\eqref{eq:phi0} acts as
	\begin{equation}\label{eq:phi3}
		\mathsf{\Phi}(\X)=\blockmatrix{\phi(X_\ee)}{B\ket{\xi_\e}}{\bra{\xi_\e}B^\dag}{\gamma x_\gg+\omega(X_\ee)},
	\end{equation}
	where now $\omega:\hilbe\rightarrow\mathbb{C}$ is a linear functional. By Theorem~\ref{thm:cp}, this map is completely positive if and only if $\omega$ is a positive functional and $\phi-\gamma^{-1}B(\cdot)B^\dag$ is completely positive, the latter condition being, in turn, equivalent to $B$ belonging to the ball $\mathcal{B}_{\phi,\gamma}$ associated with $\phi$ as defined in Eq.~\eqref{eq:bphi}. As a particular feature of the case $d_\g=1$, we are indeed able to characterize all choices of $\phi$ and $B$ rendering $\mathsf{\Phi}$ positive:
	
	\begin{proposition}\label{prop:p}
		Let $\phi:\bhe\rightarrow\bhe$, $\omega:\bhe\rightarrow\mathbb{C}$, $B\in\bhe$, and $\gamma\geq0$; let $\mathsf{\Phi}$ as in Eq.~\eqref{eq:phi3}. The following statements are equivalent:
		\begin{itemize}
			\item[(i)] $\mathsf{\Phi}$ is positive;
			\item[(ii)] the map $\omega$ is positive, and one of the following conditions hold:
			\begin{itemize}
				\item[$\bullet$] $\gamma=0$ and $B=0$;
				\item[$\bullet$] $\gamma>0$, and the map $\phi-\frac{1}{\gamma}B(\cdot)B^\dag$ is positive.
			\end{itemize}
		\end{itemize}
	\end{proposition}
	\begin{proof}
		The case $\gamma=0$ is again obvious, so fix $\gamma>0$. First of all notice that, by convex linearity, $\mathsf{\Phi}$ is positive if and only if
		\begin{equation}\label{eq:schur}
			\forall\ket
			{\xi}\in\hilbe,\,\forall c\in\mathbb{C}:\quad \mathsf{\Phi}\blockmatrix{\ketbra{\xi}{\xi}}{c\ket{\xi}}{\bra{\xi}c^*}{|c|^2}=
			\blockmatrix{\phi(\ketbra{\xi}{\xi})}{cB\ket{\xi}}{\bra{\xi}c^*B^\dag}{\gamma|c|^2+\omega(\ketbra{\xi}{\xi})}
			\succeq0.
		\end{equation}
		(i)$\implies$(ii): by Eq.~\eqref{eq:schur} for $c=0$, necessarily $\phi$ and $\omega$ must both be positive maps. Besides, for $c\neq0$, Eq.~\eqref{eq:schur} is equivalent to
		\begin{equation}\label{eq:ineq1}
			\forall\ket
			{\xi}\in\hilbe,\,\forall 0\neq c\in\mathbb{C}:\quad\phi(\ketbra{\xi}{\xi})-\frac{|c|^2}{\gamma|c|^2+\omega(\ketbra{\xi}{\xi})}B\ketbra{\xi}{\xi}B^\dag\succeq0,
		\end{equation}
		the denominator being nonzero since $\gamma>0$ and $\omega$ is a positive map. This must hold for all $c$; in particular, taking the limit $|c|\to\infty$ and recalling that the strong limit of positive semidefinite operators is positive semidefinite, we end up with the condition
		\begin{equation}\label{eq:ineq2}
			\forall\ket
			{\xi}\in\hilbe:\quad\phi(\ketbra{\xi}{\xi})-\frac{1}{\gamma}B\ketbra{\xi}{\xi}B^\dag\succeq0,
		\end{equation}
		which, again by convex linearity, is equivalent to the positivity of $\phi-\frac{1}{\gamma}B(\cdot)B^\dag$.
		
		(ii)$\implies$(i): by assumption, Eq.~\eqref{eq:ineq2} holds, which clearly implies Eq.~\eqref{eq:ineq1}, and thus Eq.~\eqref{eq:schur}, for all $c\neq0$; finally, Eq.~\eqref{eq:schur} for $c=0$ follows from the positivity of $\phi$ (which is an immediate consequence of the positivity of $\phi-\gamma^{-1}B(\cdot)B^\dag$) and $\omega$.
	\end{proof}
	
	Proposition~\ref{prop:p}, together with Theorem~\ref{thm:cp}, allows us to conclude that, as long as we choose beforehand $B$ as a combination of the Kraus operators associated with $\phi$, then complete positivity and positivity are equivalent properties. Such a property, which was already observed for the (multilevel) amplitude-damping and phase-damping channels~\cite{Davide-1,Davide-2}, is therefore a much more general feature.
	\begin{corollary}\label{coroll:cp=p}
		Let $\phi:\bhe\rightarrow\bhe$ with Kraus representation $\phi=\sum_{\mu=1}^rA_\mu(\cdot)A_\mu^\dag$, $\omega:\bhe\rightarrow\bhg$, $\gamma\geq0$, and $B\in\mathrm{Span}\,\{A_1,\dots,A_r\}$; let $\mathsf{\Phi}$ as in Eq.~\eqref{eq:phi3}. The following statements are equivalent:
		\begin{itemize}
			\item[(i)] $\mathsf{\Phi}$ is completely positive;
			\item[(ii)] $\mathsf{\Phi}$ is positive.
		\end{itemize}
	\end{corollary}
	\begin{proof}
		Clearly (i)$\implies$(ii). To prove the converse implication, suppose that $\mathsf{\Phi}$ is not completely positive; by Theorem~\ref{thm:cp}, this happens if and only if either $\omega$ is not a positive functional or, expressing $B$ as $\sum_{\mu=1}^r\beta_\mu A_\mu$ for some $\beta_1,\dots,\beta_r\in\mathbb{C}$, we have $\sum_{\mu=1}^r|\beta_\mu|^2>\gamma$. In the first case, there is $0\preceq X_\ee\in\bhe$ such that $\omega(X_\ee)<0$, and then
		\begin{equation}
			\mathsf{\Phi}\blockmatrix{X_\ee}{0}{0}{0}=\blockmatrix{\phi(X_\ee)}{0}{0}{\omega(X_\ee)},
		\end{equation}
		so that $\mathsf{\Phi}$ transforms a positive semidefinite operator into an operator with a negative eigenvalue; therefore, $\mathsf{\Phi}$ is positive. In the second case, we have
		\begin{equation}
			\phi-\frac{1}{\gamma}B(\cdot)B^\dag=\sum_{\mu,\nu=1}^r\left[\delta_{\mu\nu}-\frac{\beta_\mu\beta_\nu^*}{\gamma}\right]A_\mu(\cdot)A_\nu^\dag
		\end{equation}
		and it is easy to show that the map above is positive if and only if $\sum_{\mu=1}^r|\beta_\mu|^2\leq\gamma$, which is false by assumption.
	\end{proof}
	In particular, the equivalence between positivity and complete positivity \textit{always} holds whenever $\phi$ has maximal Kraus rank $r=d_\e^2$, since in such a case $\mathrm{Span}\{A_1,\dots,A_r\}=\bhe$.
	
	Finally, we point out that the equivalence between positivity and complete positivity fails in the case $d_\g>1$, as can be immediately shown by taking into account the trivial case $\phi=B=0$: any positive $\omega$ which is not completely positive will render $\mathsf{\Phi}$ positive but, likewise, not completely positive. This counterexample does not apply if $d_\g=1$ since, in such a case, positivity and complete positivity coincide.
	
	\subsection{Case \texorpdfstring{$d_\e=d_\g=1$}{de=dg=1}}
	
	Let us finally analyze briefly the qubit case, i.e.~$d_\e=d_\g=1$. Choosing any couple $\ket{\e}\in\hilbe$, $\ket{\g}\in\hilbg$ of states with unit norm, the most general $\X\in\hilb$ can be written as
	\begin{equation}
		\X=\blockmatrix{x_\ee\ketbra{\e}{\e}}{x_\eg\ketbra{\e}{\g}}{x_\ge\ketbra{\g}{\e}}{x_\gg\ketbra{\g}{\g}}\simeq\blockmatrix{x_\ee}{x_\eg}{x_\ge}{x_\gg}
	\end{equation}
	for some $x_\ee,x_\eg,x_\ge,x_\gg\in\mathbb{C}$, with the obvious isomorphisms $\hilbe,\hilbg\simeq\mathbb{C}$ being henceforth understood. With this representation, the most general map $\mathsf{\Phi}:\bh\rightarrow\bh$ as in Eq.~\eqref{eq:phi0} can be written as follows:
	\begin{equation}
		\mathsf{\Phi}(\X)=\blockmatrix{|a|^2x_\ee}{bx_\eg}{b^*x_\ge}{\gamma x_\gg+|q|^2x_\ee}
	\end{equation}
	for some $a,b,q\in\mathbb{C}$, $\gamma\geq0$. This family of maps contains some familiar qubit channel, cf.~Eq.~\eqref{eq:qubit}:
	\begin{itemize}
		\item when $\gamma=1$, $|b|=|a|$ and $|q|^2=1-|a|^2$, it reduces to an amplitude-damping channel;
		\item when $\gamma=1$ and $|a|=|q|=1$, it reduces to a phase-damping channel.
	\end{itemize}
	In all such cases, $\gamma=0$, Theorem~\ref{thm:cp} and Corollaries~\ref{coroll:cp}--\ref{coroll:cp=p} imply that complete positivity \textit{and} positivity hold if and only if $|b|\leq\sqrt{\gamma}|a|$; also, the map is CPTP if and only if $\gamma=1$ and $|q|^2=1-|a|^2$, cf.~Example~\ref{ex:tp} in the qubit case.
	
	\section{Time-dependent scenario: semigroups and CP-divisibility} \label{II}	
	We shall now analyze time-dependent excitation-damping channels. Given $\hilbe,\hilbg,\hilb$ as before, we shall consider a time-dependent map $\mathsf{\Phi}_t:\bh\rightarrow\bh$ defined as follows: for all $t\geq0$ and all $\X\in\bh$ partitioned as in Eq.~\eqref{eq:part}, we set
	\begin{equation}\label{eq:phit0}
		\mathsf{\Phi}_t(\X)=\blockmatrix{\phi_t(X_\ee)}{B_t X_\eg}{X_\ge B_t^\dag}{\gamma_t X_\gg+\omega_t(X_\ee)},
	\end{equation}
	where $\phi_t:\bhe\rightarrow\bhe$, $\omega_t:\bhe\rightarrow\bhg$, $B_t\in\bhe$, and $\gamma_t\geq0$. Since we will be ultimately interested in the case of \textit{trace preserving} time-dependent maps, we shall hereafter restrict our attention to the case $\gamma_t=1$ in Eq.~\eqref{eq:phit0}, so that
	\begin{equation}\label{eq:phit}
		\mathsf{\Phi}_t(\X)=\blockmatrix{\phi_t(X_\ee)}{B_t X_\eg}{X_\ge B_t^\dag}{X_\gg+\omega_t(X_\ee)}.
	\end{equation}
	Theorem~\ref{thm:cp} implies that $\mathsf{\Phi}_t$ is completely positive at all $t\geq0$ if and only if $\omega_t$ and $\phi_t-B_t(\cdot)B_t^\dag$ are completely positive at all times; besides, by Proposition~\ref{prop:inv}, $\mathsf{\Phi}_t$ is invertible at all times if and only if $\phi_t,B_t$ are all invertible at all times, with
	\begin{equation}\label{eq:phitinv}
		\mathsf{\Phi}_t^{-1}(\X)=\blockmatrix{\phi_t^{-1}(X_\ee)}{B_t^{-1}X_\eg}{X_\ge B_t^{\dag-1}}{X_\gg-\omega_t\circ\phi_t^{-1}(X_\ee)}.
	\end{equation}
	
	By using Theorem~\ref{thm:cp}, we will be able characterize all excitation-damping semigroups in Subsection~\ref{subsec:semigroup}, and then all invertible CP-divisible excitation-damping channel in Subsection~\ref{subsec:cpdiv}.
	
	\subsection{Excitation-damping Markovian semigroups}\label{subsec:semigroup}	
	Recall that a time-dependent map $\mathsf{\Phi}_t:\bh\rightarrow\bh$ is said to satisfy the semigroup property whenever
	\begin{equation}
		\mathsf{\Phi}_t\circ\mathsf{\Phi}_s=\mathsf{\Phi}_{t+s}\qquad\forall t,s\geq0;
	\end{equation}
	in particular, all completely positive maps satisfying said property, together with $\mathsf{\Phi}_0=\mathsf{id}$, are usually referred to as Markovian semigroups~\cite{ALICKI}. Besides, all trace non-increasing semigroups with $\mathsf{\Phi}_0=\mathsf{id}$ can be characterized via the relation $\mathsf{\Phi}_t=\exp(t\mathsf{L})$, where $\mathsf{L}:\bh\rightarrow\bh$, the Gorini--Kossakowski--Lindblad--Sudarshan (GKLS) generator, reads, for all $\X\in\bh$,
	\begin{equation}
		\mathsf{L}(\X)=\mathsf{E}(\X)-\left[\mathsf{\Gamma X+X\Gamma^\dag}\right],
	\end{equation}
	where $\mathsf{E}:\bh\rightarrow\bh$ is a completely positive map, and $\mathsf{\Gamma}\in\bh$ with $\mathrm{Re}\,\mathsf{\Gamma}\succeq0$. By choosing a Kraus representation $\mathsf{E}=\sum_\mu\mathsf{F}_\mu(\cdot)\mathsf{F}_\mu^\dag$ for the former, and writing $\mathsf{\Gamma}$ without loss of generality as
	\begin{equation}
		\mathsf{\Gamma}=\i\mathsf{H}+\frac{1}{2}\left(\mathsf{G}+\sum_{\mu}\mathsf{F_\mu^\dag F_\mu}\right),\quad \mathsf{H}=\mathsf{H}^\dag,\;\;\mathsf{G}=\mathsf{G}^\dag\succeq0,
	\end{equation}
	so that $\tr\mathsf{L}(\X)=\tr(\mathsf{GX})$, we restore the familiar expression for GKLS generators:
	\begin{equation}
		\mathsf{L}(\X)=-\i[\mathsf{H},\X]-\frac{1}{2}\{\mathsf{G},\X\}+\sum_{\mu}\left[\mathsf{F}_\mu\X\mathsf{F}_\mu^\dag-\frac{1}{2}\{\mathsf{F}^\dag_\mu \mathsf{F}_\mu,\X\}\right],
	\end{equation}
	where the symbols $[\cdot,\cdot]$ and $\{\cdot,\cdot\}$ are respectively used for the commutator and anticommutator. By construction, the corresponding semigroup is trace preserving, and hence CPTP, if and only if $\tr\mathsf{L}(\X)=0$ for all $\X\in\bh$, which happens if and only $\tr\mathsf{GX}=0$, and thus if and only if $\mathsf{G}=\mathsf{0}$; in this sense, $\mathsf{G}$ is the operator responsible for the loss of trace (cf. Remark~\ref{rem:ww}).
	
	The first main result of this section follows: a complete characterization of all completely positive excitation-damping semigroups, and in particular all CPTP excitation-damping semigroups, can be obtained.	
	\begin{theorem}\label{thm:semigroups}
		For all $t\geq0$, let $\phi_t:\bhe\rightarrow\bhe$, $\omega_t:\bhe\rightarrow\bhg$, and $B_t\in\bhe$, with all functions being continuously differentiable; let $\mathsf{\Phi}_t:\bh\rightarrow\bh$ as in Eq.~\eqref{eq:phit}. The following statements are equivalent:
		\begin{itemize}
			\item[(i)] $\mathsf{\Phi}_t$ is a completely positive semigroup;
			\item[(ii)] the following conditions hold:
			\begin{itemize}
				\item[$\bullet$] $\phi_t=\exp(t L)$, where $L:\bhe\rightarrow\bhe$ is a GKLS generator, that is,
				\begin{equation}\label{eq:gkls}
					L=-\i[H,\cdot]-\frac{1}{2}\{G,\cdot\}+\sum_{\mu=1}^r\left(F_\mu(\cdot)F_\mu^\dag-\frac{1}{2}\{F_\mu^\dag F_\mu,\cdot\}\right)
				\end{equation}
				for some $H,G,F_1,\dots,F_r\in\bhe$, with $H=H^\dag$ and $G=G^\dag\succeq0$;
				\item[$\bullet$] $B_t=\exp(tK)$, where
				\begin{equation}\label{eq:kappa}
					K=-\i H-\frac{1}{2}\left(G+\sum_{\mu=1}^rF_\mu^\dag F_\mu\right)-\left(\i\varepsilon+\frac{\kappa}{2}\right)\oper_\e-\sqrt{\kappa}\,\sum_{\mu=1}^rc_\mu F_\mu
				\end{equation}
				for some $\varepsilon\in\mathbb{R}$, $\kappa\geq0$ and some $c_1,\dots,c_r\in\mathbb{C}$ with $\sum_{\mu=1}^r|c_\mu|^2\leq1$;
				\item[$\bullet$] $\omega_t$ is given by
				\begin{equation}\label{eq:omegat}
					\omega_t=\psi\circ\int_0^t\mathrm{d}\tau\,\exp(\tau L)
				\end{equation}
				for some completely positive map $\psi:\bhe\rightarrow\bhg$.
		\end{itemize}\end{itemize}
		Furthermore, $\mathsf{\Phi}_t$ is a completely positive and trace preserving semigroup if and only if, in addition, $\tr L(X_\ee)+\tr\psi(X_\ee)=0$, or equivalently		
		\begin{equation}\label{eq:tr}
			\tr\psi(X_\ee)=\tr GX_\ee ,
		\end{equation}
		for all $X_\ee\in\bhe$.
	\end{theorem}
	We refer to Subsection~\ref{subsec:proof} for the proof. This theorem basically implies that, once a completely positive trace non-increasing semigroup on the excited sector $\phi_t:\bhe\rightarrow\bhe$, i.e.~$\phi_t = \e^{tL }$ with $L$ as defined in Eq.~\eqref{eq:gkls}, has been fixed, our freedom in constructing a completely positive excitation-damping semigroup $\mathsf{\Phi}_t$ consists entirely in the choice of
	\begin{itemize}
		\item the parameters $\varepsilon \in \mathbb{R}$, $\kappa \geq 0$, and $c_1,\dots,c_r\in\mathbb{C}$ with by $\sum_\mu |c_\mu|^2 \leq 1$ in Eq.~\eqref{eq:kappa};
		\item a completely positive map $\psi:\bhe\rightarrow\bhg$,
	\end{itemize}
	with the latter being further constrained by Eq.~\eqref{eq:tr} if one imposes the trace preserving condition. We stress that larger values of $\kappa$ yield a quicker decoherence phenomenon, with the coherence completely vanishing in the limit $\kappa\to\infty$.
	
	\begin{remark}\label{rem:ww}
		The case of trace preserving semigroups, i.e. those for which the additional condition~\eqref{eq:tr} holds, is of particular importance from a physical standpoint. Indeed, such a case, Theorem~\ref{thm:semigroups} characterize all possible ways to upgrade a completely positive, but generally only trace non-increasing semigroup $\phi_t$, into a CPTP semigroup $\mathsf{\Phi}_t$. In terms of master equations, this can be restated as follows: $\phi_t$ satisfies a master equation in the form
		\begin{equation}\label{eq:ww}
			\dot\phi_t(X_\ee)=-\i\left(H_{\rm eff}X_\ee-X_\ee H_{\rm eff}^\dag\right)+\sum_\mu\left[F_\mu X_\ee F_\mu^\dag-\frac{1}{2}\{F_\mu^\dag F_\mu,X_\ee\}\right]
		\end{equation}
		where
		\begin{equation}
			H_{\rm eff}=H-\frac{\i}{2}G,
		\end{equation}
		plays the role of an (effective) non-hermitian Hamiltonian; instead, $\mathsf{\Phi}_t$ satisfies
		\begin{equation}\label{eq:ww2}
			\dot{\mathsf{\Phi}}_t(\X)=-\i\left(\mathsf{H}\X-\X \mathsf{H}\right)+\sum_\mu\left[\mathsf{F}_\mu \X \mathsf{F}_\mu^\dag-\frac{1}{2}\{\mathsf{F}_\mu^\dag \mathsf{F}_\mu,\X\}\right],
		\end{equation}
		with a hermitian Hamiltonian: the enlargement of the Hilbert space allows us to describe a decay phenomenon (e.g. a particle decay) by means of a CPTP semigroup. As such, Theorem~\ref{thm:semigroups} generalizes the results derived in~\cite{Beatrix}, which are reproduced by taking $\varepsilon=\kappa=0$, $c_1=\ldots=c_r=0$ and $\psi = M(\cdot)M^\dagger$, where $G = M^\dagger M$ (cf.~Example~\ref{ex:tp3}).
	\end{remark}
	
	\begin{example}\label{ex:ww}
		In the framework depicted in Remark~\ref{rem:ww}, the simplest scenario corresponds to a Wigner--Weisskopf~\cite{WW} dynamics: in such a case, the evolution in the excited sector is governed by
		\begin{equation}\label{eq:ww1}
			\dot\phi_t(X_\ee)=-\i\left(H_{\rm eff}X_\ee-X_\ee H_{\rm eff}^\dag\right) ,
		\end{equation}
		where
		\begin{equation}
			H_{\rm eff}=H-\frac{\i}{2}G,
		\end{equation}
		plays the role of an (effective) non-hermitian Hamiltonian. One finds		
		\begin{equation}\label{}
			\phi_t(X_\ee) = A_t X_\ee A_t^\dagger ,
		\end{equation}
		with $A_t = \e^{-\i H_{\rm eff}t}$. Hence		
		\begin{equation}\label{}
			K = -\i H_{\rm eff} - \left(\i\varepsilon+\frac{\kappa}{2}\right)\oper_\e ,
		\end{equation}
		giving rise to		
		\begin{equation}\label{}
			B_t = \e^{Kt} = \e^{-\i \varepsilon t} \e^{- \kappa t/2} A_t .
		\end{equation}
		Finally,
		\begin{equation}\label{}
			\omega_t(X_\ee) = \int_0^t \mathrm{d}\tau \:\psi(A_\tau X_\ee A_\tau^\dagger)  .
		\end{equation}
		Importantly, even in this simple case, we still have a freedom on the construction of a semigroup $\mathsf{\Phi}_t$  in the choice of the parameters $\varepsilon\in\mathbb{R}$, $\kappa\geq0$ which enter the modulation of the coherence blocks (with, in particular, higher values of $\kappa$ corresponding to quicker decoherence phenomena) and especially in the choice of the map $\psi$ such that the total trace is preserved, cf.~Examples~\ref{ex:tp2}--\ref{ex:tp3}.
	\end{example}
	
	In view of Remark~\ref{rem:ww}, before proving the theorem, it will be useful to provide some examples of manifestly trace preserving excitation-damping semigroups.
	\begin{example}\label{ex:tp2}
		As an example of excitation-damping CPTP semigroup as in Eq.~\eqref{eq:phit}, fix $\Omega\in\bhg$ with $\tr\Omega=1$, and let
		\begin{equation}
			X_\ee\in\bhe\mapsto\psi(X_\ee):=-\tr L(X_\ee)\;\Omega=\left(\tr GX_\ee\right)\Omega,
		\end{equation}
		which clearly satisfies the condition for trace invariance. Correspondingly, the map $\omega_t:\bhe\rightarrow\bhg$ acts as such:
		\begin{eqnarray}
			\omega_t(X_\ee)&=&\psi\circ\int_0^t\mathrm{d}\tau\;\exp(\tau L)X_\ee=-\left[\int_0^t\mathrm{d}\tau\;\tr\left(L\exp(\tau L)X_\ee\right)\right]\Omega\nonumber\\
			&=&-\left[\int_0^t\mathrm{d}\tau\;\frac{\mathrm{d}}{\mathrm{d}\tau}\tr\left(\exp(\tau L)X_\ee\right)\right]\Omega=\tr\left[X_\ee-\exp(tL)X_\ee\right]\Omega,
		\end{eqnarray}
		which therefore belongs to the subclass of such models considered in Example~\ref{ex:tp}.
	\end{example}
	\begin{example}\label{ex:tp3}
		More generally, a possible choice to satisfy Eq.~\eqref{eq:tr}, thus obtaining a CPTP excitation-damping semigroup, is the following: writing $G = \sum_m M_m^\dagger M_m$ for some family of operators $\{M_m\}_m\subset\mathcal{B}(\hilbe,\hilbg)$, we set				
		\begin{equation}
			\psi = \mathcal{E}\circ\mathcal{M},\quad\text{where}\;\;\mathcal{M}=\sum_m M_m (\cdot) M_m^\dagger ,
		\end{equation}
		where $\mathcal{E} : \mathcal{B}(\mathcal{H}_\g) \to \mathcal{B}(\mathcal{H}_\g)$ is any completely positive and trace preserving map on the ground sector.	With this choice, clearly Eq.~\eqref{eq:tr}	is satisfied and a CPTP excitation-damping semigroup is obtained. The particular case of Example~\ref{ex:tp2} is recovered by choosing $\mathcal{E}=\tr(\cdot)\Omega$: indeed, in such a case,
		\begin{eqnarray}
			\psi(X_\ee)&=&\mathcal{E}\left(\sum_m M_mX_\ee M_m^\dag\right)=\tr\left(\sum_m M_mX_\ee M_m^\dag\right)\Omega\nonumber\\
			&=&\tr\left(\sum_m M_m^\dag M_mX_\ee \right)\Omega=\left(\tr GX_\ee\right)\Omega.
		\end{eqnarray}
	\end{example}
	
	\subsection{Proof of \texorpdfstring{Theorem~\ref{thm:semigroups}}{Theorem 3.1}}\label{subsec:proof}
	To prove Theorem~\ref{thm:semigroups}, it will be useful to start by first proving a necessary and sufficient condition for $\mathsf{\Phi}_t$ to be a (not necessarily completely positive) semigroup.
	\begin{lemma}\label{lemma:semigroups}
		Let $\mathsf{\Phi}_t$ as above, with $\phi_t$ being completely positive and all quantities being continuously differentiable as functions of $t$. Then $\mathsf{\Phi}_t$ satisfies the semigroup property and $\mathsf{\Phi}_0=\mathsf{id}$, if and only if there are a GKLS generator $L:\bhe\rightarrow\bhe$, a map $\psi:\bhe\rightarrow\bhg$, and an operator $K\in\bhe$, such that
		\begin{equation}\label{eq:easy}
			\phi_t=\exp(tL),\quad B_t=\exp(tK),
		\end{equation}
		and
		\begin{equation}\label{eq:lesseasy}
			\omega_t=\psi\circ\int_0^t\mathrm{d}\tau\,\exp(\tau L).
		\end{equation}
	\end{lemma}
	\begin{proof}
		By Eq.~\eqref{eq:phit}, a straightforward computation shows that $\mathsf{\Phi}_t\circ\mathsf{\Phi}_s=\mathsf{\Phi}_{t+s}$ for all $t,s\geq0$ and $\mathsf{\Phi}_0=\mathsf{id}$ if and only if, for all $t,s\geq0$, the following conditions hold:
		\begin{eqnarray}
			\phi_t\circ\phi_s=\phi_{t+s},&\;&\phi_0=\mathrm{id};\\
			B_tB_s=B_{t+s},&\;&B_0=\oper_\e;\\\label{eq:omega0}
			\omega_s+\omega_t\circ\phi_s=\omega_{t+s},&\;&\omega_0=0.
		\end{eqnarray}
		The first two conditions are clearly equivalent to $\phi_t,B_t$ being as in Eq.~\eqref{eq:easy}. We must prove that the latter condition~\eqref{eq:omega0} is satisfied if and only if there exists some $\psi:\bhe\rightarrow\bhg$ such that Eq.~\eqref{eq:lesseasy} holds. Now, if Eq.~\eqref{eq:lesseasy} holds, then an immediate computation shows that Eq.~\eqref{eq:omega0} holds. Vice versa, suppose that Eq.~\eqref{eq:omega0} holds for all $t,s\geq0$, that is,
		\begin{equation}
			\omega_{t+s}-\omega_s+\omega_s=\omega_t\circ\exp(sL).
		\end{equation}
		Dividing by $t$ and taking the limit $t\to0$, we get the following differential equation for $\omega_s$:
		\begin{equation}
			\dot\omega_s=\dot\omega_0\circ\exp(sL),
		\end{equation}
		with the initial condition $\omega_0=0$, which is uniquely solved by Eq.~\eqref{eq:lesseasy} with $\psi=\dot\omega_0$.
	\end{proof}
	
	\begin{proof}[Proof of Theorem~\ref{thm:semigroups}]
		(ii)$\implies$(i) Let $\phi_t=\exp(tL)$ with $L:\bhe\rightarrow\bhe$ being a GKLS generator on $\bhe$, $B_t=\exp(tK)$, and $\omega_t$ as in Eq.~\eqref{eq:omegat}. Lemma~\ref{lemma:semigroups} ensures that $\mathsf{\Phi}_t$ is a semigroup; it is therefore a Markovian semigroup if and only if can be written as $\mathsf{\Phi}_t=\exp(t\mathsf{L})$, with $\mathsf{L}:\bh\rightarrow\bh$ being a GKLS generator on $\bh$.
		
		By differentiating Eq.~\eqref{eq:phit}, for all $t\geq0$ we get $\dot{\mathsf{\Phi}}_t(\X)=\mathsf{L}\!\left(\mathsf{\Phi}_t(\X)\right)$, where
		\begin{equation}
			\mathsf{L}(\X)=\blockmatrix{L(X_\ee)}{KX_\eg}{X_\ge K^\dag}{\psi(X_\ee)}.
		\end{equation}
		Now, since $L$ is a GKLS generator, it admits the following decomposition:
		\begin{equation}
			L(X_\ee)=E(X_\ee)-\left[\Gamma X_\ee+X_\ee\Gamma^\dag\right],
		\end{equation}
		with $E:\bhe\rightarrow\bhe$ completely positive, and $\Gamma\in\bhe$ with $\Re\Gamma\succeq0$. A simple computation then shows that
		\begin{equation}
			\mathsf{L}(\X)=\mathsf{E}(\X)-\left[\mathsf{\Gamma}\X+\X\mathsf{\Gamma}^\dag\right],
		\end{equation}
		where
		\begin{eqnarray}\label{eq:e}
			\mathsf{E}(\X)&=&\blockmatrix{E(X_\ee)}{\left(K+\Gamma+\i\varepsilon+\frac{\kappa}{2}\right)X_\eg}{X_\eg\left(K+\Gamma+\i\varepsilon+\frac{\kappa}{2}\right)^\dag}{\kappa X_\gg+\psi(X_\ee)},\\
			\mathsf{\Gamma}&=&\blockmatrix{\Gamma+\i\varepsilon}{0}{0}{\frac{\kappa}{2}\oper_\g},
		\end{eqnarray}
		with the map in Eq.~\eqref{eq:e} being completely positive by Theorem~\ref{thm:cp} as long as $\psi$ is completely positive, $\kappa\geq0$, and the operator $K+\Gamma+\left(\i\varepsilon+\frac{\kappa}{2}\right)\oper_\g$ is a linear combination of any family of Kraus operators $F_1,\dots,F_r$ associated with the map $E:\bhe\rightarrow\bhe$:
		\begin{equation}
			K+\Gamma+\left(\i\varepsilon+\frac{\kappa}{2}\right)\oper_\e=\sum_{\mu=1}^r\tilde{c}_\mu F_\mu,\qquad\sum_{\mu=1}^r|\tilde{c}_\mu|^2\leq\kappa,
		\end{equation}
		or equivalently, defining $c_\mu:=\tilde{c}_\mu/\sqrt{\kappa}$,
		\begin{equation}
			K=-\Gamma-\left(\i\varepsilon+\frac{\kappa}{2}\right)\oper_\e+\sqrt{\kappa}\,\sum_{\mu=1}^rc_\mu F_\mu,\qquad\sum_{\mu=1}^r|c_\mu|^2\leq1.
		\end{equation}
		The claim therefore follows by identifying $\Gamma$ with the operator $\i H+\frac{1}{2}\left(G+\sum_{\mu=1}^rF_\mu^\dag F_\mu\right)$.
		
		(i)$\implies$(ii). First of all, in order for $\mathsf{\Phi}_t$ to be completely positive, by Theorem~\ref{thm:cp} necessarily $\phi_t$ must be completely positive for all $t$. Besides, since $\mathsf{\Phi}_t$ is a semigroup, by Lemma~\ref{lemma:semigroups} necessarily we must have
		\begin{equation}
			\phi_t=\exp(tL),\qquad B_t=\exp(tK),\qquad\omega_t=\psi\circ\int_0^t\mathrm{d}\tau\;\exp(\tau L)
		\end{equation}
		for some GKLS generator $L:\bhe\rightarrow\bhe$ as in Eq.~\eqref{eq:gkls}, some $K\in\bhe$, and some map $\psi:\bhe\rightarrow\bhg$. Again by Theorem~\ref{thm:cp}, $\omega_t$ must be completely positive at all times, which implies that $\psi$ is completely positive.
		
		Finally, Theorem~\ref{thm:cp} requires $\phi_t-B_t(\cdot)B_t^\dag$, and thus $\phi_t-B_t(\cdot)B_t^\dag$ to be a completely positive operator. Now, since $\phi_t$ is a Markovian semigroup with GKLS generator given by Eq.~\eqref{eq:gkls}, the map $\phi_t$ admits a family of Kraus operators $A_0(t),A_1(t),\dots,A_r(t)$ satisfying
		\begin{eqnarray}
			A_0(t)&=&\oper_\e-t\left[\i H+\frac{1}{2}\left(G+\sum_{\mu=1}^rF_\mu^\dag F_\mu\right)\right]+\mathcal{O}(t^2);\\
			A_\mu(t)&=&\sqrt{t}\,F_\mu+\mathcal{O}(t),\;\;\mu=1,\dots,r,
		\end{eqnarray}
		which, in particular, imply for all $t>0$
		\begin{eqnarray}\label{eq:diff}
			\dot A_0(t)&=&-\left[\i H+\frac{1}{2}\left(G+\sum_{\mu=1}^rF_\mu^\dag F_\mu\right)\right]+\mathcal{O}(t);\\\label{eq:diff2}
			\dot A_\mu(t)&=&\frac{1}{2\sqrt{t}}F_\mu+\mathcal{O}(1),\;\;\mu=1,\dots,r.
		\end{eqnarray}
		Since $\phi_t-B_t(\cdot)B_t^\dag$ must be completely positive, recalling Lemma~\ref{lemma:ball} we must have, for all $t\geq0$,
		\begin{equation}\label{eq:cond}
			B_t=\exp(tK)=\sum_{\mu=0}^r\beta_\mu(t)A_\mu(t),\qquad\sum_{\mu=0}^r|\beta_\mu(t)|^2\leq1
		\end{equation}
		for some $\beta_0(t),\dots,\beta_r(t)\in\mathbb{C}$, with the right-hand side of Eq.~\eqref{eq:cond} necessarily being differentiable. Eq.~\eqref{eq:cond} at $t=0$ implies $\beta_0(0)=1$ and $\beta_\mu(0)=0$ for all $\mu=1,\dots,r$. Besides, differentiating Eq.~\eqref{eq:cond}, we get
		\begin{equation}
			K\exp(tK)=\dot\beta_0(t)A_0(t)+\beta_0(t)\dot A_0(t)+\sum_{\mu=1}^r\left[\dot\beta_\mu(t)A_\mu(t)+\beta_\mu(t)\dot A_\mu(t)\right],
		\end{equation}
		and, as $t\to0$,
		\begin{equation}\label{eq:cond2}
			K=\lim_{t\to0}\dot\beta_0(t)\oper_\e-\left[\i H+\frac{1}{2}\left(G+\sum_{\mu=1}^rF_\mu^\dag F_\mu\right)\right]+\sum_{\mu=1}^r\lim_{t\to0}\left[\sqrt{t}\dot\beta_\mu(t)+\frac{\beta_\mu(t)}{2\sqrt{t}}\right]F_\mu.
		\end{equation}
		Necessarily,
		\begin{equation}
			\lim_{t\to0}\dot\beta_0(t)=:-\i\varepsilon-\frac{\kappa}{2}
		\end{equation}
		is finite, and we must have $\beta_\mu(t)=c_\mu\sqrt{t}+o(\sqrt{t})$ for some $c_\mu\in\mathbb{C}$. Finally, the condition $\sum_{\mu=0}^r|\beta_\mu(t)|^2\leq1$ at small times implies
		\begin{equation}
			\sum_{\mu=1}^r|c_\mu|^2\leq\kappa
		\end{equation}
		with $\kappa\geq0$; Eq.~\eqref{eq:kappa} follows by rescaling the $c_\mu$s properly, finally proving the claim.
	\end{proof}
	
	\subsection{CP-divisible excitation-damping channels}\label{subsec:cpdiv}	
	Let us recall that a dynamical map $\mathsf{\Phi}_t$ is divisible whenever, for all $t \geq s\geq0$, the following decomposition	
	\begin{equation}
		\mathsf{\Phi}_t = \mathsf{\Phi}_{t,s}\circ\mathsf{\Phi}_s ,
	\end{equation}
	holds true for some family of maps (\textit{propagators}) $\mathsf{\Phi}_{t,s}:\bh\rightarrow\bh$; if so, one calls $\mathsf{\Phi}_t$ CP-divisible if the family of propagators $\mathsf{\Phi}_{t,s}$ is CPTP, and P-divisible if $\mathsf{\Phi}_{t,s}$ is positive and trace preserving.
	
	We shall restrict ourselves to the case of \textit{invertible} processes: such processes are divisible, with $\mathsf{\Phi}_{t,s} = \mathsf{\Phi}_{t}\circ\mathsf{\Phi}_{s}^{-1}$. In the case of excitation-damping maps, we know from Proposition~\ref{prop:inv} that $\mathsf{\Phi}_t$ is invertible if and only if $\phi_t$ and $B_t$ are invertible, with
	\begin{equation}\label{eq:phinv0}
		\mathsf{\Phi}_t^{-1}(\X)=\blockmatrix{\phi_t^{-1}(X_\ee)}{B_t^{-1}X_\eg}{X_\ge B_t^{\dag-1}}{X_\gg-\omega_t\circ \phi_t^{-1}(X_\ee) };
	\end{equation}
	in such a case $\mathsf{\Phi}_{t}$ satisfies a time-local master equation $\dot{\mathsf{\Phi}}_t=\mathsf{L}_t\circ\mathsf{\Phi}_t$, and the corresponding time-local generator $\mathsf{L}_t$ is defined via
	\begin{equation}
		\mathsf{L}_t=\dot{\mathsf{\Phi}}_t\circ\mathsf{\Phi}_t^{-1};
	\end{equation}
	a straightforward computation leads to	
	\begin{equation}\label{eq:blockgkls}
		\mathsf{L}_t(\X) = \blockmatrix{L_t(X_\ee)}{K_tX_\eg}{X_\ge K_t^\dag}{\psi_t(X_\ee)},
	\end{equation}
	where $L_t:\bhe\rightarrow\bhe$, $K_t\in\bhe$, and $\psi_t:\bhe\rightarrow\bhg$ are defined via
	\begin{eqnarray}\label{eq:omegat2}
		L_t&=&\dot{\phi}_t\circ\phi_t^{-1};   \label{Lt}\\
		K_t&=&\dot B_tB_t^{-1};    \label{Kt}\\
		\psi_t&=& \dot\omega_t\circ\phi_t^{-1}.
	\end{eqnarray}
	Such relations can be (formally) inverted as follows:
	\begin{eqnarray}
		\phi_t = \mathcal{T}\exp\left(\int_0^t\mathrm{d}\tau\;L_\tau\right),\qquad B_t = \mathcal{T}\exp\left(\int_0^t\mathrm{d}\tau\;K_\tau\right),
	\end{eqnarray}
	with $\mathcal{T}$ denoting the time-ordered product. Solving Eq.~\eqref{eq:omegat2} for $\omega_t$, one finds
	\begin{equation}\label{p1}
		\omega_t=\int_0^t \mathrm{d}\tau\; \psi_\tau\circ\phi_\tau.
	\end{equation}
	The evolution is CP-divisible if and only if the family of propagators $\mathsf{\Phi}_{t,s} = \mathsf{\Phi}_t\circ \mathsf{\Phi}_s^{-1}$ is CPTP for all pairs $t \geq s$. One finds
	\begin{equation}
		\mathsf{\Phi}_{t,s}(\X)=\blockmatrix{\phi_{t,s}(X_\ee)}{B_{t,s} X_\eg}{X_\ge B_{t,s}^{\dag}}{X_\gg + \omega_{t,s}(X_\ee) };
	\end{equation}
	where	
	\begin{eqnarray}
		\phi_{t,s} &=& \phi_t\circ\phi_s^{-1} = \mathcal{T}\exp\left(\int_s^t\mathrm{d}\tau\;L_\tau\right)\ ; \\
		B_{t,s} &=& B_t B_s^{-1} = \mathcal{T}\exp\left(\int_s^t\mathrm{d}\tau\; K_\tau\right),
	\end{eqnarray}
	and	
	\begin{equation}
		\omega_{t,s} = (\omega_t - \omega_s)\circ\phi_s^{-1} .
	\end{equation}
	Using Eq.~\eqref{p1} one obtains	
	\begin{equation}
		\omega_{t,s} = \int_s^t \mathrm{d}\tau\, \psi_\tau\circ\phi_{\tau,s}  =  \int_s^t \mathrm{d}\tau\, \psi_\tau \circ\left\{ \mathcal{T}\exp\left(\int_s^\tau \mathrm{d} u\;L_u\right) \right\} \ ,
	\end{equation}
	which, in the homogeneous case, reduces to	
	\begin{equation}
		\omega_{t - s} = \int_s^t \mathrm{d}\tau\, \psi\circ\phi_{\tau- s}  = \psi\circ\int_0^{t-s} \mathrm{d}\tau\,  \phi_{\tau} = \psi\circ\int_0^{t-s} \mathrm{d}\tau\, \exp\left(\tau L\right) .
	\end{equation}
	Consequently, we obtain the following generalization of Theorem~\ref{thm:semigroups}:	
	\begin{theorem}\label{thm:cpdiv}
		For all $t\geq0$, let $\phi_t:\bhe\rightarrow\bhe$, $\omega_t:\bhe\rightarrow\bhg$, and $B_t\in\bhe$, with all functions being continuously differentiable; let $\mathsf{\Phi}_t:\bh\rightarrow\bh$ as in Eq.~\eqref{eq:phit}. The following statements are equivalent:
		\begin{itemize}
			\item[(i)] $\mathsf{\Phi}_t$ is a CP-divisible dynamical map,
			\item[(ii)] the following conditions hold:
			\begin{itemize}
				\item[$\bullet$] $L_t$, defined in Eq.~\eqref{Lt}, is a trace non-increasing GKLS generator for all $t \geq 0$, that is,
				\begin{equation}\label{eq:gkls2}
					L_t =-\i[H_t,\cdot]-\frac{1}{2}\{G_t,\cdot\}+\sum_{\mu=1}^r\left(F_{\mu,t}(\cdot)F_{\mu,t}^\dag-\frac{1}{2}\{F_{\mu,t}^\dag F_{\mu,t},\cdot\}\right) ,
				\end{equation}
				for some time dependent operators $H_t,G_t,F_{1,t},\dots,F_{r,t}\in\bhe$, with $H_t=H_t^\dag$ and $G_t =G_t^\dag\succeq0$;
				\item[$\bullet$] $K_t$, defined in Eq.~\eqref{Kt}, reads
				\begin{equation}\label{eq:kappa2}
					K_t =-\i H_t -\frac{1}{2}\left(G_t +\sum_{\mu=1}^rF_{\mu,t}^\dag F_{\mu,t} \right)-\left(\i\varepsilon_t +\frac{\kappa_t}{2}\right) \oper_\e-\sqrt{\kappa_t}\,\sum_{\mu=1}^rc_{\mu,t} F_{\mu,t}
				\end{equation}
				for some $\varepsilon_t\in\mathbb{R}$, $\kappa_t\geq0$ and some $c_{1,t},\dots,c_{r,t}\in\mathbb{C}$ satisfying $\sum_{\mu=1}^r|c_{\mu,t}|^2\leq1$;
				\item[$\bullet$] $\omega_t$ is defined in Eq.~\eqref{p1} for some completely positive map $\psi_t :\bhe\rightarrow\bhg$.
		\end{itemize}\end{itemize}
		Furthermore, $\mathsf{\Phi}_t$ is a completely positive and trace preserving channel if and only if, in addition, $\tr L_t(X_\ee)+\tr\psi_t(X_\ee)=0$, or equivalently		
		\begin{equation}\label{eq:tr2}
			\tr\psi_t(X_\ee)=\tr G_tX_\ee ,
		\end{equation}
		for all $X_\ee\in\bhe$ and all $t\geq0$.
	\end{theorem}
	The proof goes exactly on the same line as the proof of Theorem~\ref{thm:semigroups}. As in the semigroup case, our freedom in constructing a CP-divisible excitation-damping map resides in the choice of
	\begin{itemize}
		\item the time-dependent parameters $\varepsilon_t \in \mathbb{R}$, $\kappa_t \geq 0$, and $c_{1,t},\dots,c_{r,t}\in\mathbb{C}$ with by $\sum_\mu |c_{\mu,t}|^2 \leq 1$ in Eq.~\eqref{eq:kappa};
		\item a time-dependent completely positive map $\psi_t:\bhe\rightarrow\bhg$,
	\end{itemize}
	with the latter being further constrained by Eq.~\eqref{eq:tr2} if one imposes the trace preserving condition. Some final considerations are in order.
	
	\begin{remark}\label{rem:ww-t}
		Reprising the discussion in Remark~\ref{rem:ww}, let us focus again on the trace preserving scenario. Again, in such a case the result can be restated as follows: the map $\phi_t$, which satisfies a generalization of Eq.~\eqref{eq:ww}:
		\begin{equation}\label{eq:ww-t}
			\dot\phi_t(X_\ee)=-\i\left(H_{\rm{eff},t}X_\ee-X_\ee H_{\rm{eff},t}^\dag\right)+\sum_\mu\left[F_{\mu,t} X_\ee F_{\mu,t}^\dag-\frac{1}{2}\{F_{\mu,t}^\dag F_{\mu,t},X_\ee\}\right]
		\end{equation}
		with a \textit{time-dependent} effective non-hermitian Hamiltonian $H_{\rm{eff},t}$, is upgraded to a map $\mathsf{\Phi}_t$ satisfying a generalization of Eq.~\eqref{eq:ww2}:
		\begin{equation}\label{eq:ww2-t}
			\dot{\mathsf{\Phi}}_t(\X)=-\i\left(\mathsf{H}_t\X-\X \mathsf{H}_t\right)+\sum_\mu\left[\mathsf{F}_{\mu,t} \X \mathsf{F}_{\mu,t}^\dag-\frac{1}{2}\{\mathsf{F}_{\mu,t}^\dag \mathsf{F}_{\mu,t},\X\}\right],
		\end{equation}
		with a \textit{time-dependent} hermitian Hamiltonian $\mathsf{H}_t$.
	\end{remark}
	
	\begin{example}\label{ex:tp4}			
		A simple scenario ensuring trace preservation corresponds again to the choices $\varepsilon_t=\kappa_t=0$, $c_{1,t}=\ldots=c_{r,t}=0$, and $\psi_t(X_\ee) = \mathcal{E}_t(M_t X_\ee M_t^\dagger)$, where $G_t=M_t^\dag M_t$ and $\mathcal{E}_t:\bhg\rightarrow\bhg$ is a family of trace preserving maps such that				
		\begin{equation}
			\omega_t(X_\ee) = \int_0^t \mathrm{d}\tau\: \mathcal{E}_\tau (M_\tau \phi_\tau(X_\ee) M^\dagger_\tau) ,
		\end{equation}
		is completely positive. In this case, one finds
		\begin{eqnarray}
			\dot{X}_{\ee,t} &=& -\i[H_t,X_{\ee,t}]-\frac{1}{2}\{G_t,X_{\ee,t}\}+\sum_{\mu=1}^r\left(F_{\mu,t}X_{\ee,t} F_{\mu,t}^\dag-\frac{1}{2}\{F_{\mu,t}^\dag F_{\mu,t},X_{\ee,t}\}\right)   \ , \\
			\dot{X}_{\eg,t} &=& \left[ -\i H_t -\frac{1}{2}\left(G_t +\sum_{\mu=1}^rF_{\mu,t}^\dag F_{\mu,t} \right)\right] X_{\eg,t}
		\end{eqnarray}
		and
		\begin{equation}
			X_{\gg,t} = \int_0^t \mathrm{d}\tau \:\mathcal{E}_\tau (M_\tau X_{\ee,\tau} M^\dagger_\tau) ,
		\end{equation}
		Clearly, whenever $\mathcal{E}_\tau $ is completely positive, then the map				
		\begin{equation}
			\psi_\tau(X_\ee) = \mathcal{E}_\tau (M_\tau X_\ee M^\dagger_\tau)
		\end{equation}
		is completely positive and hence $\mathsf{\Phi}_t$ is a CP-divisible dynamical map. However, when $\mathcal{E}_\tau $ is not completely positive but $\omega_t$ is completely positive for all $t\geq0$, then  $\mathsf{\Phi}_t$ is a legitimate map but not CP-divisible.
		
		This shows that a CP-divisible trace non-increasing map $\phi_t$ can be extended to a non-CP-divisible trace preserving map $\mathsf{\Phi}_t$. If we identify the CP-divisibility property as a mathematical description of Markovianity, this provides an example of a non-Markovian quantum system which, nevertheless, behaves as a Markovian one on the excited sector. Likewise, when $L_t=L$ and $K_t=K$ are time independent and hence characterized by Theorem~\ref{thm:semigroups}, one may still have a time dependent $\psi_t$; therefore, a semigroup $\phi_t$ can be extended to a CP-divisible evolution which, however, fails to satisfy the semigroup property. 	
		
		Finally, generalizing the discussion in Example~\ref{ex:tp3}, a CP-divisible and trace preserving excitation-damping map can be obtained as follows: writing $G_t = \sum_m M_{m,t}^\dagger M_{m,t}$ for some time-dependent family of operators $\{M_{m,t}\}_m\subset\mathcal{B}(\hilbe,\hilbg)$, we set
		\begin{equation}
			\psi_t = \mathcal{E}_t\circ\mathcal{M}_t,\quad\text{where}\;\;\mathcal{M}_t=\sum_m M_{m,t} (\cdot) M_{m,t}^\dagger ,
		\end{equation}
		with $\mathcal{E}_t:\mathcal{B}(\mathcal{H}_\g)\to\mathcal{B}(\mathcal{H}_\g)$ being an arbitrary quantum channel (a CPTP map).
	\end{example}
	
	\section{Conclusions and outlooks}\label{III}	
	In this work we have introduced and studied a class of quantum operations describing excitation-damping phenomena, that is, one-way transfer of probability between two sectors of the Hilbert space of a quantum system, an excited and a ground one. Such maps may be interpreted as generalizations the well-known amplitude-damping channel, as well as of their extension to multilevel excited sectors studied in~\cite{Davide-1}, and can be considered as a way of upgrading a trace non-increasing map (quantum operation) to a possibly trace preserving map on a larger Hilbert space, by adding additional degrees of freedom playing the role of a ground sector.
	
	For such maps, a simple characterization of complete positivity has been obtained; besides, in the case of one-dimensional ground sectors, complete positivity turns out to be equivalent to the (generally much weaker) positivity. This nontrivial feature generalizes what already observed for amplitude-damping and phase-damping qubit channels, as well as their generalizations studied respectively in~\cite{Davide-1} and~\cite{Davide-2}, thus showing that the equivalence between the two properties does indeed hold for a wider class of maps.
	
	Furthermore, the time-dependent scenario has been examined; a characterization of all dynamical semigroups, and more generally all invertible and CP-divisible maps, belonging to the class of excitation-damping channels has been obtained. This characterization involves an explicit expression for all admissible generators, which must be carefully chosen in such a way to ensure complete positivity at all times.
	
	Future research will be devoted to investigating other mathematical and physical properties of excitation-damping maps, as well as their practical implementation, and to study the time-dependent case via a memory kernel approach.
	
	\section*{Acknowledgments}	
	D.L. was partially supported by Istituto Nazionale di Fisica Nucleare (INFN) through the project “QUANTUM” and by the Italian National Group of Mathematical Physics (GNFM-INdAM), and acknowledges support by MIUR via PRIN 2017 (Progetto di Ricerca di Interesse Nazionale), project QUSHIP (2017SRNBRK); he also thanks the Institute of Physics at the Nicolaus University in Toru\'n for its hospitality. D.C. was supported by the Polish National Science Center Project No. 2018/30/A/ST2/00837.

\end{document}